\documentclass[11pt,a4paper]{article}

\usepackage[utf8]{inputenc}
\usepackage{amssymb,amsfonts,amsmath,amsthm}
\usepackage{appendix}
\usepackage{color}
\usepackage{enumerate}
\usepackage{graphicx}             
\usepackage{cite}
\usepackage{multirow}
\usepackage{a4wide}
\usepackage[english]{babel}
\renewcommand\Pr{{\mathbf P}}
\newcommand{\mscript}[1]{{\mbox{\scriptsize #1}}}
\DeclareMathOperator{\Li}{Li}

\newtheorem{theorem}{Theorem}
\newtheorem{lemma}{Lemma}
\newtheorem{definition}{Definition}

\long\def\symbolfootnote[#1]#2{\begingroup%
\def\thefootnote{\fnsymbol{footnote}}\footnotetext[#1]{#2}\endgroup}

\newcommand{\vc}[1]{\mathbf{#1}}
\newcommand{\mx}[1]{\mathbf{#1}}

\begin{document}

\title{Controlling edge dynamics in complex networks}
\author{Tam\'as Nepusz\textsuperscript{1,2} and Tam\'as Vicsek\textsuperscript{1,2,*}}

\maketitle

\footnotetext[1]{Department of Biological Physics, Eötvös Loránd University, Pázmány Péter sétány 1/a, 1117 Budapest, Hungary.}
\footnotetext[2]{Statistical and Biological Physics Research Group of the Hungarian Academy of Sciences, Pázmány Péter sétány 1/a, 1117 Budapest, Hungary.}
\symbolfootnote[1]{Corresponding author: \texttt{vicsek@hal.elte.hu}}

\begin{abstract}
The interaction of distinct units in physical, social, biological and
technological systems naturally gives rise to complex network structures.
Networks have constantly been in the focus of research for the last decade,
with considerable advances in the description of their structural and dynamical
properties.  However, much less effort has been devoted to studying
the controllability of the dynamics taking place on them.
Here we introduce and evaluate a dynamical process defined on the
edges of a network, and demonstrate that the controllability properties of this
process significantly differ from simple nodal dynamics. Evaluation of
real-world networks indicates that most of them are more controllable than
their randomized counterparts. We also find that transcriptional regulatory
networks are particularly easy to control. Analytic calculations show that
networks with scale-free degree distributions have better controllability
properties than uncorrelated networks, and positively correlated in- and
out-degrees enhance the controllability of the proposed dynamics.
\end{abstract}

The last decade has witnessed an explosive growth of interest in the
descriptive analysis of complex natural and technological systems that
permeate many aspects of everyday life\cite{newman03,albert02,boccaletti06}.
Research in network science has mostly been focused on
measuring\cite{watts98,barabasi99,amaral00},
modeling\cite{newman02b,leskovec05} and decomposing
\cite{fortunato10,newman04,palla05} network representations of existing natural
phenomena in order to deepen our understanding of the underlying systems.
Considerably less attention has been dedicated to the various types of
network dynamics\cite{ebel02,luscombe04,prill05,palla07b} and even less to the problem
of controllability\cite{yu09,lombardi07,rahmani09}, i.e. determining the
conditions under which the dynamics of a network can be driven from any initial
state to any desired final state within finite
time\cite{kalman63,sontag98,slotine91,liu11}.

Structural controllability\cite{lin74} has been proposed recently as a
framework for studying the controllability properties of directed complex
networks \cite{liu11}. In this framework, a linear time-invariant nodal
dynamics is assumed on the network, governed by the following equation:
\begin{equation}
\dot{\vc{x}}(t) = \mx{A} \vc{x}(t) + \mx{B} \vc{u}(t)
\end{equation}
where $\mx{A}$ is the transpose of the (weighted) adjacency matrix of the
network, $\vc{x}(t)$ is a time-dependent vector of the state variables of the
nodes, $\vc{u}(t)$ is the vector of input signals, and $\mx{B}$ is the
so-called \emph{input matrix} which defines how the input signals are connected
to the nodes of the network. The dynamics is said to be \emph{structurally
controllable} if there exists a matrix $\mx{A}^\ast$ with the same structure
as $\mx{A}$ such that the network can be driven from any initial state to any
final state by appropriately choosing the input signals $\vc{u}(t)$\cite{lin74}.
Here, structural equivalence of $\mx{A}$ and $\mx{A}^\ast$ means that
$\mx{A}^\ast$ is not allowed to contain a non-zero entry when the corresponding
entry in $\mx{A}$ is zero. Structural controllability is a general property
in the sense that almost all weight combinations of a given network are
controllable if the network is structurally controllable for a given
$\mx{B}$\cite{lin74,shields76}.  The minimum number of input signals is then
determined by finding a maximum matching in the network, i.e. a maximum subset
of edges such that each node has at most one inbound and at most one outbound
edge from the matching. The number of nodes without inbound edges from the
matching is then equal to the number of input signals required for structural
controllability\cite{liu11}.

Perhaps the most striking feature of the structural controllability approach to
linear nodal dynamics is that input signals tend to control the hubs of the
network only indirectly. In addition, real-world networks that seem to have
evolved to control an underlying process (such as transcriptional regulatory
networks) need many input signals\cite{liu11}. This is due to the fact that
driven nodes (i.e.  those which receive an input signal directly) are not able
to control their subordinates independently from each other. However, these
results apply only for linear nodal dynamics. In this paper, we examine and
describe a dynamics that takes place on the edges of the network, and show that
this dynamics leads to significantly different controllability properties for
the same real-world networks.

\section{Switchboard dynamics in complex networks}

We study a dynamical process on the edges of a directed
complex network $G(V, E)$ as follows.  Let $\vc{x} = \left[ x_j \right]$
denote the state vector of the process, where one state variable corresponds
to each edge of the network. Let $\vc{y}^-_i$ and $\vc{y}^+_i$
be vectors consisting of those $x_j$ values that correspond to the inbound
and outbound edges of vertex $i$, respectively, and let $\mx{M}_i$ denote a
matrix with the number of rows being equal to the out-degree and the number of
columns being equal to the in-degree of vertex $i$. Furthermore, we assume that
the dynamics can be influenced from the environment by adding an offset vector
$\vc{u}_i$ to the state vector of the outbound edges of any node $i$.  The
equations governing the dynamics of the network are then as follows:
\begin{equation}
\dot{\vc{y}}^+_i(t) = \mx{M}_i \vc{y}^-_i(t) - \boldsymbol{\tau}_i \otimes \vc{y}^+_i(t)
 + \sigma_i \vc{u}_i(t)
\label{eq:dyn-eq-original}
\end{equation}
where $\boldsymbol{\tau}_i$ is a vector of damping terms corresponding to the
edges in $\vc{y}^+_i(t)$, $\sigma_i$ is 1 if vertex $i$ is a so-called
\emph{driver node} and zero otherwise, and $\otimes$ denotes the entry-wise
product of two vectors of the same size.

We call the above the \emph{switchboard dynamics} (SBD) since each vertex $i$ acts as a
small switchboard-like device mapping the signals of the inbound
edges to the outbound edges using a linear operator $\mx{M}_i$, which is called
the \emph{mixing} or \emph{switching matrix} from now on.
To simplify the equations, state variables and signals like 
$\vc{y}_i^+$, $\vc{y}_i^-$ and $\vc{u}_i$ are implicitly considered as
time-dependent, even if the time variable $t$ is omitted. Furthermore, note
that for an edge $v \to w$, exactly one of the coordinates of $\vc{u}_v$
affects the state of this edge, therefore we can simply introduce a unified
input vector $\vc{u}$ where the $j$th element $u_j$ is simply the component of
the offset vectors that affects edge $j$ directly.

In some sense, the SBD provides a simplified representation of
the underlying dynamic processes of many real-world networks. For instance, in
social communication networks, a node (i.e.~a person) is constantly processing
the information received via its inbound edges and makes decisions which are then
communicated to other nodes via the outbound edges. The inbound and outbound
signals are then represented by the state variables $x_j$, while the decision
process is modeled by the mixing matrices $\mx{M}_i$.

We must also explain the motivation of introducing the offset vectors as a means
of controlling the system instead of assuming external input signals.
In most networks, one usually can not take control over a single edge as the
connections do not always have a physical realization. Therefore, in order to control
an edge in a network, one has to take control over the vertex from which the
edge originates, and adjust the output vector of the vertex appropriately. This
adjustment is represented by the term $\sigma_i \vc{u}_i$ for each vertex $i$.
Throughout this paper, we will be interested in determining an optimal control
configuration for the SBD of a given network, where optimality is measured by
the number of driver nodes $\sigma = \sum_i \sigma_i$.

First, we make a connection between the switchboard dynamics and a standard
linear dynamical system by re-writing the equations of the
switchboard dynamics (Eq.~\eqref{eq:dyn-eq-original}) in terms of $x_i$.
Note that the derivative of the state of an arbitrary edge $j$ originating in
some vertex $r$ and terminating in vertex $s$ depends only on itself and on the
states of edges whose head is $r$. Let us denote this latter set by
$\Gamma_j^-$, simplifying our dynamical equation to
\begin{equation}
\dot{x}_j = \sum_{k \in \Gamma_j^-} w_{kj} x_k - \tau_j x_j + \sigma_s u_j
\end{equation}
where $w_{kj}$ is the element in the mixing matrix $\mx{M}_r$ of vertex $r$
that corresponds to edge $k$ (as inbound edge) and edge $j$ (as outbound
edge), $\tau_j$ is the damping term related to edge $j$, and $u_j$ is equal
to the value of the input signal affecting the state variable of edge $j$.
Defining $w_{kj} = 0$ for all $k \notin \Gamma_j^-$ yields
\begin{equation}
\dot{\vc{x}} = (\mx{W} - \vc{T}) \vc{x} + \mx{H} \vc{u}
\label{eq:dyn-eq-rewritten}
\end{equation}
where the unknown variables are as follows:
\begin{itemize}
\item $\mx{W} = \left[ w_{kj} \right]$ is a matrix where $w_{kj}$ may
be nonzero if and only if the head of edge $k$ is the tail of edge $j$.
\item $\vc{T}$ is a diagonal matrix with the damping terms of each edge in the
main diagonal.
\item $\mx{H}$ is a diagonal matrix where the $j$th diagonal element is
$\sigma_s$ if vertex $s$ is the tail of edge $j$.
\end{itemize}

Eq.~\eqref{eq:dyn-eq-rewritten} essentially describes a simple linear
time-invariant dynamical system of the form $\dot{\vc{x}} = \mx{A} \vc{x} +
\mx{B} \vc{u}$ with the substitution $\mx{A} = \mx{W} - \mx{T}$ and $\mx{B} =
\mx{H}$. It is also easy to see that $\mx{W}$ is the adjacency matrix of the
\emph{line digraph} $L(G)$ of the original digraph $G$ by definition. The nodes
of $L(G)$ thus correspond to the edges of the original network $G$, and each
edge of $L(G)$ represents a length-two directed path of $G$. An example network
$G$ is shown in Figure~\ref{fig:ex}a, and its corresponding line digraph on
Figure~\ref{fig:ex}b. The loop edges arising from the damping term $-\mx{T}$ in
Eq.~\eqref{eq:dyn-eq-original} are omitted from
Figures~\ref{fig:ex}b and \ref{fig:ex}c, partly for sake of clarity, and partly
because soon we will demonstrate that such edges do not change the optimal
control configuration.

\begin{figure}[t!]
\centering
\includegraphics{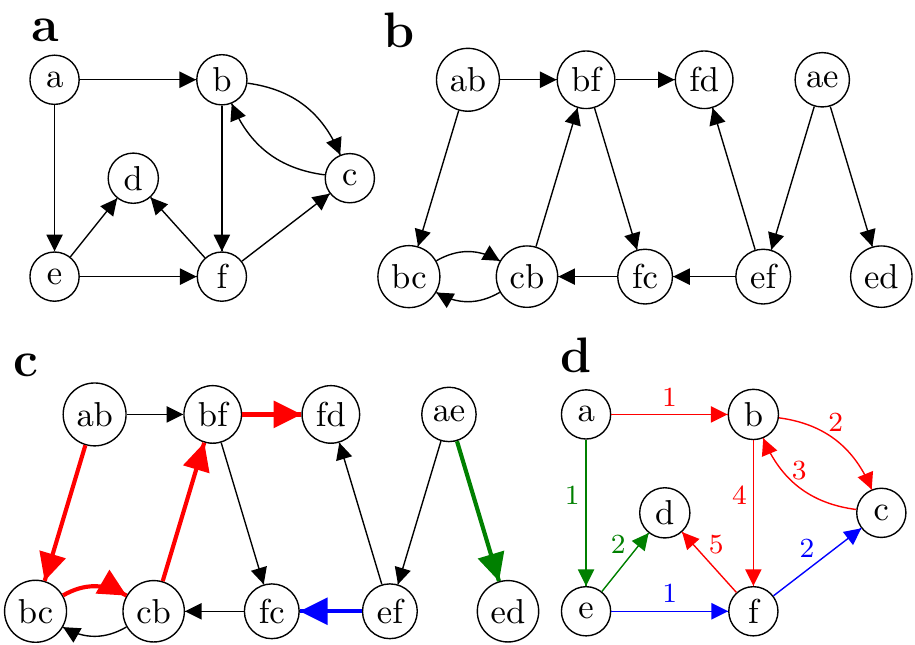}
\caption{{\bf (a)} An example network $G$ with six vertices and nine edges. The switchboard
dynamics takes place on the edges of the network. {\bf (b)} The line graph $L(G)$
corresponding to $G$. A linear time-invariant dynamics on the vertices of this
network is equivalent to the switchboard dynamics on $G$. Node labels refer to
the endpoints of the edges in $G$ to which they correspond. {\bf (c)} Applying
the maximum matching theorem to $L(G)$ yields disjoint control paths.
{\bf (d)} The control paths in $G$, mapped back from $L(G)$. Note how each path in
$L(G)$ became an edge-disjoint walk in $G$. Numbers represent the order in which
the edges have to be traversed in the walks. The two driver nodes are $a$ and $e$
since each walk starts from either $a$ or $e$.}
\label{fig:ex}
\end{figure}

\section{Structural controllability of the switchboard dynamics}

Applying the maximum matching theorem of Liu et al \cite{liu11} to $L(G)$
(Figure~\ref{fig:ex}b) gives us a set of control paths and driven nodes in the
line digraph (Figure~\ref{fig:ex}c), or equivalently, a set of driven
\emph{edges} in the original graph $G$. Since edges can be controlled only via
the offset vectors, the set of driver nodes are given by collecting those
vertices that have at least one outbound driven edge. However, note that the
maximum matching theorem guarantees only that the number of driven nodes
in $L(G)$ will be minimal, and this does not imply that the obtained set of
driver nodes in $G$ is also minimal.

Let us now compare the control paths obtained from the maximum matching in the
line graph $L(G)$ in Figure~\ref{fig:ex}c with the corresponding control paths
in the original graph $G$ in Figure~\ref{fig:ex}d. It can be seen that
the maximum matching consists of vertex-disjoint open and closed paths
(also called \emph{stem}s and \emph{bud}s) in $L(G)$, and mapping these
paths back to $G$ yields \emph{edge}-disjoint open and closed walks in $G$.
The walks together form a complete cover of the edges of $G$. Since the
first vertex of each stem has to be driven in $L(G)$, the driver nodes in
$G$ are those from which the corresponding open edge-disjoint walks
originate. Our goal is thus to find a cover that minimizes the number of nodes
from which open walks originate in $G$.

Let us call a vertex $v$ \emph{divergent} if $d_v^+ > d_v^-$, \emph{convergent}
if $d_v^+ < d_v^-$, and \emph{balanced} if $d_v^+ = d_v^-$, and let us define
a \emph{balanced component} as a connected component consisting solely of
balanced vertices and at least one edge. Our key result (which can also be
formulated as a theorem) is that the minimum set of driver nodes required to
control the SBD on a network $G(V, E)$ can be determined by selecting the
divergent vertices of $G$ and one arbitrary vertex from each balanced
component. The formal proof is given in the Appendix.

The above theorem has two important implications. First, it explains why
we are safe to ignore loop edges in $L(G)$: a loop edge of a vertex in
$G$ increases both its in-degree and its out-degree by one, thus a divergent
vertex stays divergent, and a non-divergent vertex stays non-divergent.
Second, the theorem shows that the number of driver nodes required to control
the SBD is almost completely determined by the joint degree distribution of the
network. This is in concordance with the results of Liu et al \cite{liu11} for
the linear time-invariant nodal dynamics.

\section{Controllability of real networks}

\begin{table}[tp!]
\caption{Controllability properties of the real networks analysed in this paper}
\footnotesize
\begin{tabular}{lrcccccccc} \hline \\[-8pt]
Type & \# & Name & Nodes & Edges & $n_D^\mscript{SBD}$ & $n_D^\mscript{Liu}$ & $n_D^\mscript{ER}$ & $n_D^\mscript{Degree}$ \\[2pt] \hline \\[-5pt]
Regulatory & 1. & Ownership-USCorp & 7,253 & 6,726 & 0.160 & 0.820 & 0.339 & 0.085 \\
 & 2. & TRN-EC-2 & 418 & 519 & 0.222 & 0.751 & 0.366 & 0.148 \\
 & 3. & TRN-Yeast-1 & 4,441 & 12,873 & 0.034 & 0.965 & 0.415 & 0.033 \\
 & 4. & TRN-Yeast-2 & 688 & 1,079 & 0.177 & 0.821 & 0.381 & 0.137 \\[5pt]
Trust & 5. & College$^\ast$ & 32 & 96 & 0.344 & 0.188 & 0.418 & 0.315 \\
 & 6. & Epinions$^\ast$ & 75,888 & 508,837 & 0.336 & 0.549 & 0.445 & 0.448 \\
 & 7. & Prison$^\ast$ & 67 & 182 & 0.403 & 0.134 & 0.411 & 0.451 \\
 & 8. & Slashdot$^\ast$ & 82,168 & 948,464 & 0.323 & 0.045 & 0.458 & 0.392 \\
 & 9. & WikiVote$^\ast$ & 7,115 & 103,689 & 0.281 & 0.666 & 0.463 & 0.620 \\[5pt]
Food web & 10. & Grassland & 88 & 137 & 0.318 & 0.523 & 0.381 & 0.297 \\
 & 11. & Little Rock & 183 & 2,494 & 0.639 & 0.541 & 0.463 & 0.649 \\
 & 12. & Seagrass & 49 & 226 & 0.449 & 0.265 & 0.436 & 0.433 \\
 & 13. & Ythan & 135 & 601 & 0.304 & 0.511 & 0.432 & 0.337 \\[5pt]
Metabolic & 14. & \emph{C. elegans} & 1,173 & 2,864 & 0.182 & 0.302 & 0.409 & 0.309 \\
 & 15. & \emph{E. coli} & 2,275 & 5,763 & 0.182 & 0.382 & 0.409 & 0.309 \\
 & 16. & \emph{S. cerevisiae} & 1,511 & 3,833 & 0.185 & 0.329 & 0.409 & 0.313 \\[5pt]
Electronic & 17. & s208a & 122 & 189 & 0.451 & 0.238 & 0.381 & 0.431 \\
circuits & 18. & s420a & 252 & 399 & 0.456 & 0.234 & 0.385 & 0.440 \\
 & 19. & s838a & 512 & 819 & 0.459 & 0.232 & 0.381 & 0.442 \\[5pt]
Neuronal & 20. & \emph{C. elegans} & 297 & 2,359 & 0.549 & 0.165 & 0.449 & 0.499 \\
and brain & 21. & Macaque & 45 & 463 & 0.333 & 0.022 & 0.446 & 0.457 \\[5pt]
Citation & 22. & arXiv-HepPh$^\ast$ & 34,546 & 421,578 & 0.356 & 0.232 & 0.459 & 0.577 \\
 & 23. & arXiv-HepTh$^\ast$ & 27,770 & 352,807 & 0.359 & 0.216 & 0.460 & 0.569 \\[5pt]
WWW & 24. & Google & 15,763 & 171,206 & 0.670 & 0.337 & 0.457 & 0.612 \\
 & 25. & Polblogs & 1,490 & 19,090 & 0.509 & 0.471 & 0.460 & 0.501 \\
 & 26. & nd.edu & 325,729 & 1,497,134 & 0.271 & 0.677 & 0.433 & 0.301 \\
 & 27. & stanford.edu & 281,904 & 2,312,497 & 0.665 & 0.317 & 0.450 & 0.653 \\[5pt]
Internet & 28. & p2p-1 & 10,876 & 39,994 & 0.334 & 0.552 & 0.425 & 0.344 \\
 & 29. & p2p-2 & 8,846 & 31,839 & 0.344 & 0.578 & 0.423 & 0.344 \\
 & 30. & p2p-3 & 8,717 & 31,525 & 0.343 & 0.577 & 0.424 & 0.344 \\[5pt]
Social & 31. & Twitter$^{\ast\dagger}$ & 41.7 $\times$ 10$^6$ & 1.47 $\times$ 10$^9$ & 0.402 & -- & 0.476 & 0.434 \\
communication & 32. & UCIOnline & 1,899 & 20,296 & 0.216 & 0.323 & 0.456 & 0.375 \\
 & 33. & WikiTalk & 2,394,385 & 5,021,410 & 0.022 & 0.968 & 0.399 & 0.026 \\[5pt]
Organizational & 34. & Consulting$^\ast$ & 46 & 879 & 0.522 & 0.043 & 0.458 & 0.460 \\
 & 35. & Freemans-1$^\ast$ & 34 & 645 & 0.412 & 0.088 & 0.441 & 0.476 \\
 & 36. & Freemans-2$^\ast$ & 34 & 830 & 0.588 & 0.029 & 0.439 & 0.465 \\
 & 37. & Manufacturing$^\ast$ & 77 & 2,228 & 0.597 & 0.013 & 0.468 & 0.424 \\
 & 38. & University$^\ast$ & 81 & 817 & 0.519 & 0.012 & 0.451 & 0.532 \\[5pt]
\hline
\end{tabular}
Notations are as follows: fraction of driver nodes under the switchboard dynamics
($n_D^\mscript{SBD}$) and the simple nodal dynamics \cite{liu11} ($n_D^\mscript{Liu}$);
fraction of driver nodes under the switchboard dynamics in randomized networks
using the Erdős--Rényi model ($n_D^\mscript{ER}$) and the degree-preserving configuration
model ($n_D^\mscript{Degree}$). Note that this latter model does not preserve
the joint degree distribution. Results for null models are averaged from 100
randomizations.  Networks where the edges were reversed compared to the
original publication are marked by $\ast$ (see Appendix, section \ref{sec:reversal}).
Results calculated directly from the degree distribution (i.e. not
taking into account balanced components) are marked by $\dagger$.
\label{table:real}
\end{table}

We have determined the set of driver nodes under the switchboard dynamics for
38 real networks classified into 11 categories and compared the fraction of
driver nodes $n_D$ with the model of Liu et al \cite{liu11} and with its
expected value after different types of randomizations (see
Table~\ref{table:real}).
A striking difference between the switchboard dynamics and the model of
Liu et al \cite{liu11} can be seen for two classes of
networks. Regulatory networks such as the transcriptional regulatory network of
\emph{E.coli} (TRN-EC-2 \cite{milo02}) and \emph{S.cerevisiae} (TRN-Yeast-1
\cite{balaji06}, TRN-Yeast-2 \cite{milo02}) and the ownership network of US
telecommunications and media corporations (Ownership-USCorp \cite{norlen02})
turned out to be well-controllable under the switchboard dynamics but they are
very hard to control in the linear nodal dynamics. This can be explained by the
fundamental difference between the two models. In the linear nodal dynamics,
a driven node may not influence its subordinates independently of each other,
thus the presence of out-hubs in a network degrades its controllability
significantly. In the switchboard dynamics, out-hubs behave the opposite way,
allowing one to control many state variables with a single out-hub. It follows
that driver nodes \emph{prefer} out-hubs in the SBD, while they are shown
to \emph{avoid} hubs in the linear nodal dynamics. Therefore, hubs have an
important role not only in maintaining the connectivity of a network in case
of random failures \cite{albert00,cohen00,jeong01} and containing epidemic
spreading \cite{pastorsatorras01,pastorsatorras02}, but they also make it
possible to control the network efficiently with a smaller number of driver
nodes.

The other class of networks with the largest difference between the two models
is the case of intra-organizational networks \cite{cross04,freeman79,nepusz08}.
In the model used by Liu et al, all these networks can be controlled by at most
three nodes. On closer examination, it turns out that 75\%-80\% of the
connections in each of these networks is reciprocal, i.e. an edge exists
between vertices $A$ and $B$ in both directions. A reciprocal edge pair can
easily form a bud in a maximum matching, requiring no driver node on its own,
therefore high reciprocity in a network always implies a low fraction of
driver nodes in the linear nodal dynamics, while this is not necessarily true
for the SBD.

Comparing the fraction of driver nodes for the SBD with the randomized variants
reveals that in most cases, the fraction of driver nodes required to control a
random Erdős--Rényi network\cite{erdos60,bollobas01} of the same size is larger
than the fraction of driver nodes for the real-world network, suggesting that
the structure of these networks is at least partially optimized for
controllability.  Notable exceptions are the electronic circuits \cite{milo02},
the neural network of \emph{C.elegans} \cite{achacoso92,watts98}, most of the
World Wide Web networks \cite{palla07,adamic05,albert99,leskovec05}, and the
intra-organizational networks \cite{cross04,freeman79,nepusz08}. Preserving the
in- and out-degree distributions (but not the joint distribution) brings the
fraction of driver nodes closer to the observed one after randomization, and
keeping the joint degree distribution makes the fraction of driver nodes
practically the same up to a difference of $\pm 0.002$ in the networks we have
studied, confirming that the effect of balanced components on the fraction of
driver nodes is indeed negligible for large real-world networks. Edge
deletion experiments (see Appendix) also indicate that the optimal control
configurations in the studied networks are robust to single link failures as
the networks mostly remain controllable with the same number of driver nodes
after the removal of a single edge.

\section{Analytical results for model networks}

The dependence of $n_D$ on the joint degree distribution allows us to derive
analytical formulae for the expected fraction of driver nodes for a wide
variety of model networks (see Appendix for the exact derivations).
For Erd\H{o}s--R\'enyi digraphs\cite{erdos60,bollobas01} with $n$ vertices and
an edge probability of $p$, $n_D$ is given as follows:
\begin{equation}
n_D^\mscript{ER} = \frac{1}{2} - \frac{e^{-2 \left< k \right>}}{2} I_0(2 \left< k \right>)
\end{equation}
where $\left< k \right> = np$ is the average in- and out-degree and
$I_\alpha(x)$ is the modified Bessel function of the first kind. The function
converges rapidly to 0.5 as $\left< k \right>$ increases.  Similar results are
obtained for graphs with independent exponential in- and out-degree
distributions $C e^{-k / \kappa}$ where $\kappa = 1 / \log \frac{1 + \left< k
\right>}{\left< k \right>}$:
\begin{equation}
n_D^\mscript{exp} =
\frac{\left< k \right>}{2 \left< k \right> + 1}
\end{equation}
which also approaches 0.5 rapidly as $\left< k \right> \to \infty$
(Figure~\ref{fig:sim}a).  For power-law distributed
digraphs\cite{caldarelli07,barabasi99} with $\Pr(d_v^+ = k) = \Pr(d_v^- = k) =
C k^{-\gamma} e^{-k/\kappa}$, $n_D$ is given by
\begin{equation}
n_D^\mscript{power} =
\frac{1}{2} - \frac{\Li_{2\gamma}(e^{-2/\kappa})}{2 \Li_\gamma(e^{-1/\kappa})^2}
\end{equation}
where $\Li_s(z)$ is the base $s$ polylogarithm function. As $\kappa \to \infty$,
this converges to
\begin{equation}
n_D^\mscript{power} = \frac{1}{2} - \frac{\zeta(2\gamma)}{2\zeta(\gamma)^2}
\label{eq:nd-powerlaw}
\end{equation}
(where $\zeta(x)$ is the Riemann zeta function) in the absence of any exponential
cutoff (Figure~\ref{fig:sim}b). The Appendix also contains the analytical
treatment of $k$-regular networks.

\begin{figure}[t!]
\centering
\includegraphics[width=\textwidth]{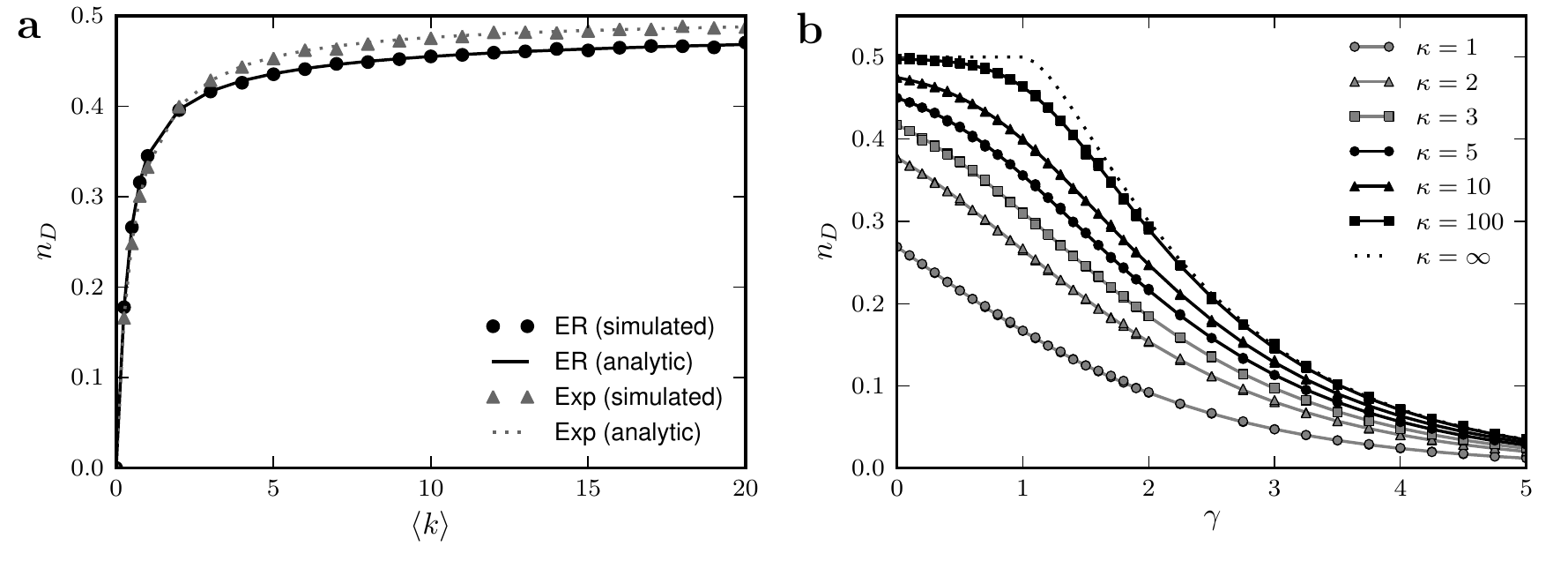}
\caption{{\bf (a)} Expected fraction of driver nodes $n_D$ in Erdős--Rényi (ER)
and exponential (Exp) networks as a function of the average degree $\left< k
\right>$.  {\bf (b)} Expected fraction of driver nodes $n_D$ in scale-free
networks with exponential cutoff as function of the exponent $\gamma$ of the
degree distribution, for different cutoff values $\kappa$.  On both panels,
symbols denote the results of simulations on networks with $10^5$ nodes, solid
lines correspond to the analytical results.}
\label{fig:sim}
\end{figure}

It is worthwhile to compare these analytical results with that of Liu et al \cite{liu11},
who have found that the fraction of driver nodes $n_D$ decreases
for both Erd\H{o}s--R\'enyi and scale-free networks as $\left< k \right> \to \infty$,
while these networks behave the opposite way under the SBD. For
$\left< k \right> \to \infty$, the fraction of driver nodes tends to $1/2$ for
Erd\H{o}s--R\'enyi networks and to $1/2 - \zeta(2\gamma) / (2 \zeta(\gamma)^2)$
for scale-free networks. The consequence is that denser networks are harder to
control (as expected by our intuition), and that scale-free networks with a
given $\left< k \right>$ are \emph{easier} to control than an
Erd\H{o}s--R\'enyi network with the same average degree. This can partly
be attributed to the higher frequency of short loops \cite{bianconi05} in
scale-free networks: these loops can be covered by closed walks and do not
require extra driver nodes.

\section{The effect of degree correlations}

Our analytical results assumed that the in-degree and the out-degree of a
node is uncorrelated, which was true for all of the model networks we have
studied. However, one-point degree correlations in real networks are significantly
different from zero \cite{bianconi08}. To study how such correlations
affect the fraction of driver nodes, we have performed simulations on
Erdős--Rényi networks and scale-free networks with an exponential cutoff and
varied the correlation as follows. First, we generated an instance of the
network model with $n = 10^5$ nodes and calculated the in- and out-degree
sequences. These instances were uncorrelated since neither the Erdős--Rényi
model nor the configuration model (which we have used to generate scale-free
networks) introduces correlations between the in- and out-degree of the same
node. Next, while keeping the in-degree sequence intact, we started swapping
elements in the out-degree sequence randomly such that only those swaps were
performed which increased the correlation. The process was continued until we
were not able to increase the correlation any more in the last $t$ steps (where
$t = 10^4$ in our simulations). A similar greedy algorithm was executed from
the original degree sequences in the opposite direction, performing swaps only
if it decreased the correlation, terminating when it was not possible to
decrease the correlation any more in the last $t$ steps. The fraction of driver
nodes $n_D$ was then calculated in the original configuration and whenever the
absolute difference of the calculated in- and out-degree correlation between
the last examined state and the current state became larger than 0.01.

The results are depicted in Figures~\ref{fig:degree_correlations}a and
\ref{fig:degree_correlations}b, both of which clearly show a general trend:
increasing the correlation between the in- and out-degrees decreases the
fraction of driver nodes. Negative one-point correlations yield a
higher fraction of driver nodes since these networks are very unlikely to
contain balanced nodes: a vertex either has a high in-degree and a low
out-degree or a high out-degree and a low in-degree.  In other words, negative
correlations indicate a clear separation of responsibilities between the nodes
of the network: divergent nodes are strongly divergent with a large difference
between the out-degree and the in-degree, while convergent nodes are strongly
convergent. Positive correlations indicate that nodes often represent complex
decision processes which map a high-dimensional input space into a similarly
high-dimensional output space. Strong positive correlations also yield networks
with a higher number of short loops \cite{bianconi08}, which can then be
covered by closed walks that do not require driver nodes on their own.

\begin{figure}[t!]
\centering
\includegraphics[width=\textwidth]{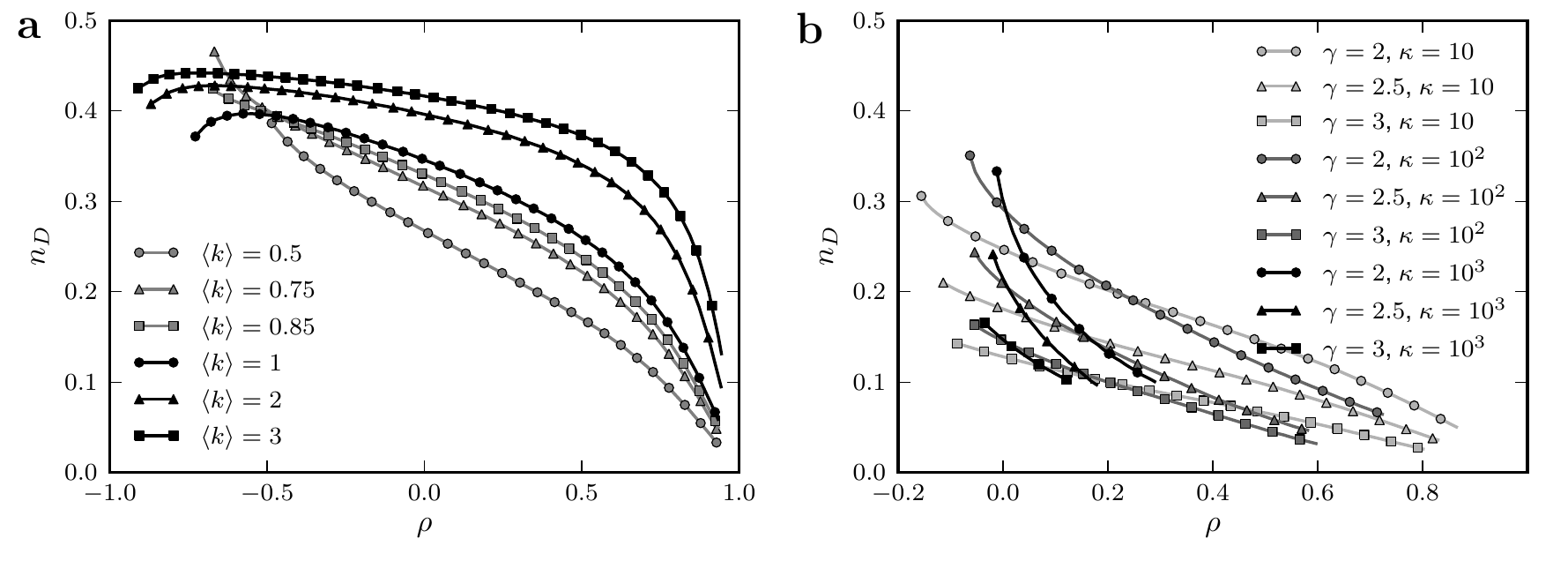}
\caption{{\bf (a)} Fraction of driver nodes $n_D$ in Erdős--Rényi (ER) networks
with different average degree $\left< k \right>$ as a function of in- and
out-degree correlation ($\rho$). {\bf (b)} Fraction of driver nodes $n_D$ in
scale-free networks with different exponents $\gamma$ as a function of in- and
out-degree correlation ($\rho$). On both panels, every fifth data point is
marked by a symbol. Each data point was obtained by averaging at least 20
different realizations of the network model; error bars were omitted as they
were smaller than the symbols.  Note that it is very hard to introduce negative
degree correlations in the case of scale-free networks and none of our test
runs managed to decrease the correlation below -0.2.
}
\label{fig:degree_correlations}
\end{figure}

\section{Conclusions}

We have presented a linear time-invariant dynamical model where state variables
correspond to the edges of a directed complex network, and the nodes of the
network act as linear operators that map state variables of inbound edges to
outbound edges.  We have demonstrated that the minimum number of driver nodes
for such systems is largely determined by the joint degree distribution of the
network. A comprehensive survey of 38 real-world networks showed that
transcriptional regulatory networks are well-controllable with a small
number of driver nodes under the switchboard dynamics, and that most
real-world networks are easier to control than random Erdős--Rényi networks
with the same number of nodes and edges. This is very different from the
findings of Liu et al \cite{liu11} who have reported a high fraction of
driver nodes under linear nodal dynamics on regulatory networks and that
randomized Erdős--Rényi networks are easier to control than the real-world
ones. The results suggest that one should choose the dynamical model
carefully before studying the controllability properties of a real-world
network as it may affect the results to a very large extent.

The behaviour of the nodal and edge dynamics is markedly different in highly
hierarchical, tree-like networks where the presence of central out-hubs rapidly
increase the required number of driver nodes for the linear nodal dynamics of
Liu et al, while the same out-hubs allow efficient control of many subordinate
nodes and thus decrease the required number of driver nodes in the switchboard
dynamics. Such hierarchies are ubiquitous in nature and society, from scales as
small as gene regulatory networks \cite{milo02,balaji06}, through leader-follower
relationships of flocking pigeons \cite{nagy10}, to the large-scale organization
of some man-made social structures like the Wikipedia talk network \cite{leskovec10}
or the ownership network of US media and telecommunications corporations \cite{norlen02}.
The presence or absence of hierarchy thus seems to be an important contributing
factor of the controllability properties of large dynamical systems.

As it happens so often in scientific research, the framework we have presented
raises more questions than answers. For instance, it is yet unknown how the
switchboard dynamics would behave in the presence of noise or nonlinearity, or
in cases when it is enough to control only a subset of the state variables
(output controllability) and only ensure that the uncontrollable ones have
stable dynamics (stabilizability).  However, as we have shown, even the first
steps along our approach could be used to deepen our understanding of the
origins of controllability of real-world networks.

\section*{Acknowledgments}

This research was supported by the EU ERC COLLMOT project. We are grateful to
Enys Mones for useful discussions.

\appendix
\appendixpage

\section{Structural controllability}

\subsection{Controllability conditions}

A continuous linear dynamical system of the form $\dot{\vc{x}} =
\mx{A} \vc{x}, \mx{A} \in \mathbb{R}^{n \times n}$ is said to be
\emph{controllable} by a set of piecewise continuous input signals $\vc{u}$ if
we are able to drive the state vector $\vc{x}$ from any arbitrary initial state
to any given state in finite time, assuming that the input signals are injected
to the linear system according to the following dynamical equation:
\begin{equation}
\dot{\vc{x}} = \mx{A} \vc{x} + \mx{B} \vc{u}
\label{eq:controllability-basic}
\end{equation}
where $\mx{B} \in \mathbb{R}^{n \times m}$ is a matrix that describes how the
input signals affect the derivatives of the state variables. $\mx{A}$ is
usually called the \emph{state matrix} and $\mx{B}$ the \emph{control matrix}.
Note that the structure of $\mx{B}$ is not constrained in any way;
one can connect any of the input signals to any of the state variables.

The \emph{Kalman rank condition} states that the system is controllable if and
only if the \emph{controllability matrix} $\left[ \mx{B} \;
\mx{A}\mx{B} \; \mx{A}^2\mx{B} \; \dots \; \mx{A}^{n-1}\mx{B} \right]$ has rank
$n$, where $n$ is the number of state variables \cite{kalman63,sontag98}. However, the
rank condition is not constructive since it does not tell us how to find an
appropriate $\mx{B}$ for a given $\mx{A}$ (preferably with a minimum number of
columns), and even testing the Kalman rank condition is computationally
expensive and numerically unstable for large $n$. Due to these difficulties,
control theorists turned to the concept of \emph{structural controllability},
first introduced by Lin in his seminal paper \cite{lin74}.

In the structural controllability framework, one assumes that the matrices
$\mx{A}$ and $\mx{B}$ contain two kinds of elements: fixed zeros and free
parameters. The free parameters of the matrices may assume any real value and
are independent of each other. A system with state matrix $\mx{A}$ and control
matrix $\mx{B}$ is then structurally controllable if it is possible to set the
free parameters of the matrices in a way that the system becomes controllable
in the usual sense. It can also be shown that if a system is structurally
controllable, then it is also controllable for all parameter values except a
set of combinations with zero Lebesgue measure \cite{lin74,shields76}; in other
words, structural controllability is a general property of the system and it
implies controllability for almost all combinations of the free parameters.

Sufficient and necessary conditions for the structural controllability of
linear time-invariant dynamical systems with known state and control matrices
$\mx{A} = \left[ a_{ij} \right]$ and $\mx{B} = \left[ b_{ij} \right]$ were
given earlier in the literature \cite{murota87,reinschke97}. The
graph-theoretic formulation of the condition is given as follows. Let $G(V, E)$
the graph representation of the dynamic system, where $V = \left\{ x_1, x_2,
\dots, x_n, u_1, u_2, \dots, u_m \right\}$, $E = E_A \cup E_B$, $E_A = \left\{
(x_i, x_j) | a_{ji} \ne 0 \right\}$ and $E_B = \left\{ (u_i, x_j) | b_{ji} \ne
0 \right\}$. Similarly, let $G^\ast(V^\ast, E^\ast)$ be the so-called
\emph{bipartite graph representation} of the system, where $V^\ast = V_R \cup
V_C \cup V_U$, $V_R = \left\{ x_1^+, x_2^+, \dots, x_n^+\right\}$,
$V_C = \left\{ x_1^-, x_2^-, \dots, x_n^- \right\}$, $V_U = \left\{ u_1, u_2,
\dots, u_m \right\}$, $E^\ast =
E^\ast_A \cup E^\ast_B$, $E^\ast_A = \left\{ (x_i^+, x_j^-) | a_{ji} \ne 0
\right\}$ and $E^\ast_B = \left\{ (u_i, x_j^-) | b_{ji} \ne 0 \right\}$.  Note
that there exists a bijection between $E$ and $E^\ast$: $(x_i, x_j)$ in $E$ is
equivalent to $(x_i^+, x_j^-)$ in $E^\ast$, and $(u_i, x_j)$ in $E$ is
equivalent to $(u_i, x_j^-)$ in $E^\ast$. The system is then structurally
controllable if and only if the following two conditions hold:
    
\begin{enumerate}
\item For all $v \in V$, $v$ is reachable from at least one of $\left\{ u_1, u_2, \dots,
      u_m \right\}$ via directed paths in $G$. This is called the
      \emph{reachability} condition. \label{reachability-cond}
\item $G^\ast$ contains $n$ independent edges, where a set of edges is
\emph{independent} if every vertex in $V^\ast$ is incident on at most one of the
edges.
\end{enumerate}

The bijection between $E$ and $E^\ast$ means that the set of independent edges
satisfying the above conditions selects $n$ edges from $E$ such that each
vertex $v \in V$ is incident on at most one inbound and at most one outbound selected
edge. For sake of simplicity and also to conform with
the terminology introduced earlier by Liu et al \cite{liu11}, we will call a set
of edges satisfying this condition a \emph{matching}\footnote{The definition of
``matching'' in this manuscript is not to be confused with matchings on
\emph{undirected} graphs, where it is required that the selected edges share no
common vertices. In this manuscript, ``matching'' always refers to a directed
matching as defined above.}, vertices with inbound selected
edges \emph{matched} and vertices without such edges \emph{unmatched}.
Note that all input vertices $u_i$ are unmatched since they have no inbound
edges, and no ordinary vertex $v_i$ will be unmatched because 1) we have
selected $n$ independent edges, 2) there are exactly $n$ ordinary vertices, 3)
we know that none of the input vertices are matched, and 4) a selected edge can
make only one vertex matched.

The selected edges form vertex-disjoint directed paths and cycles in $G$. The
directed paths are called \emph{stems} and they always originate from one of the
input vertices $u_i$ (since the first vertex of a stem is unmatched and only the
input vertices are unmatched). The directed cycles are called \emph{buds}. Stems and buds
together form the set of \emph{control paths}, since we can think about them in an
informal way as principal routes along which control signals propagate in the
system. A peculiar property of buds is that they do not require a separate
control signal: if any vertex of a bud is adjacent to the vertex of a stem,
then the stem itself will be responsible for providing the appropriate input to
the vertices of the bud as well. Note that due to the reachability condition
(see page~\pageref{reachability-cond}) there can be no bud in the system that
is not adjacent to any of the stems since each vertex is accessible from at
least one input vertex, and each input vertex is the root of one of the stems.

The above conditions can only be used to check whether the system is
structurally controllable once the control matrix $\mx{B}$ is known. In a
recent paper \cite{liu11}, Liu et al have proven the \emph{maximum matching
theorem}, which states that the minimum number of input signals required to
control a system represented by its state matrix $\mx{A}$ can be determined by
finding a maximum matching in $\mx{A}$ or a maximum set of independent edges
in the bipartite representation of $\mx{A}$
(assuming no input vertices), and then counting the number of unmatched vertices.
The maximum matching theorem also constructs the matrix $\mx{B}$ in a way that
the bipartite representation of the system with state matrix $\mx{A}$ and
control matrix $\mx{B}$ will satisfy the above conditions for structural
controllability. This is achieved by connecting an input signal to every
unmatched vertex of the graph to form one stem for each such vertex, and also
connecting these input signals to any buds that are not adjacent to stems in
order to satisfy the reachability condition. The nodes of the original network
that are connected directly to one of the input signals are called \emph{driven
nodes}, and the nodes of the input signals are called \emph{driver nodes}.

It is important to emphasize the distinction between driver and driven nodes,
since the difference between them may be arbitrarily large. The reason for this
is that one driver node may drive more than one driven node. To show that this
is not just a rare theoretical possibility, we refer the Reader to a recently
published manuscript of Cowan et al \cite{cowan11}, where the authors argue
that the dynamic equations of real networks usually include a damping term for
each state variable. These damping terms ensure that the system returns to some ground
state in the absence of external stimuli. The damping terms are represented by
nonzero diagonal elements in $\mx{A}$ and by self-loops in the network
representation of the system. When all the nodes in the network are equipped
with self-loops, a trivial maximum matching can be obtained by selecting the
self-loops only, thus constructing $n$ buds. According to the maximum matching
theorem of Liu et al, the system then requires a single driver node only, which
will be connected to all the nodes of the network.  Thus, the number of
\emph{driver} nodes will be 1 and the number of \emph{driven} nodes will be
$n$, attaining the maximum possible difference of $n-1$. Furthermore, the
theorem then states that \emph{every} real-world network with such self-loops
can be driven by a single input signal.

The distinction between driver and driven nodes is fundamental in the simple
linear time-invariant nodal dynamics assumed by Liu et al, but not in the
switchboard dynamics. In the switchboard dynamics, driver nodes are
\emph{internal} to the system: these are the nodes whose mixing matrix $M_i$
is controlled in order to drive the state variables of the edges into the
desired state. Choosing all the self-loops in the line graph $L(G)$ of the
switchboard dynamics would simply promote every node of the original network
with at least one outbound edge to a driver node. Later on in Section
\ref{main-proof}, we will prove that this is not necessarily an optimal
solution and also show a linear-time algorithm that selects the optimal driver
node configuration.

\subsection{Proof of our key result}

\label{main-proof}

For sake of clarity, we repeat some definitions from the main part of the
manuscript here.

\begin{definition}[Divergent vertex]
A vertex $v$ in a digraph $G(V, E)$ is \emph{divergent} if $d_v^+ > d_v^-$,
where $d_v^+$ is the out-degree and $d_v^-$ is the in-degree of the vertex.
\end{definition}

\begin{definition}[Convergent vertex]
A vertex $v$ in a digraph $G(V, E)$ is \emph{convergent} if $d_v^+ < d_v^-$.
\end{definition}

\begin{definition}[Balanced vertex]
A vertex $v$ in a digraph $G(V, E)$ is \emph{balanced} if $d_v^+ = d_v^-$.
\end{definition}

\begin{definition}[Balanced component]
A connected component $C \subseteq V$ in a digraph $G(V, E)$ is a
\emph{balanced component} if $v$ is balanced for every $v \in C$ and
$C$ contains at least one edge.
\end{definition}

We will also need a few more definitions and lemmas:

\begin{definition}[Edge-disjoint walk]
An edge-disjoint walk of a digraph $G(V, E)$ is a sequence of vertices
$v_0, v_1, \dots v_n$ such that $v_i \to v_{i+1}$ is a member of $E$ for every
$0 \le i < n$ and each such edge appears in the walk only once. Such a
walk is \emph{open} if $v_0 \ne v_n$ and \emph{closed} otherwise.
\end{definition}

\begin{lemma}
For every connected component $C \subseteq V$ of a digraph $G(V, E)$, exactly
one of the following three statements is true:
\begin{enumerate}
\item $C$ contains no edges.
\item $C$ contains at least one convergent and at least one divergent vertex.
\item $C$ is balanced.
\end{enumerate}
\label{lemma:conv-div}
\end{lemma}

\begin{proof}
Proving that at most one of the three statements can be true at the same time
is trivial and follows from the definitions above. To complete the proof, we
must also show that at least one of the statements must always be true. This
is done by contradiction. Suppose that there exists a connected component
$C$ in some graph $G(V, E)$ for which none of the three statements holds.
$C$ then either contains at least one convergent vertex and no divergent
vertices, or at least one divergent vertex and no convergent vertices. Both
cases are contradictory since the sum of in-degrees in any connected component
$C$ must be equal to the sum of out-degrees, and balanced vertices contribute
the same amount to both sums.
\end{proof}

\begin{lemma}
For every connected component $C \subseteq V$ of a digraph $G(V, E)$ containing
at least one edge, one of the following two statements is true:
\begin{enumerate}
\item $C$ can be covered by a single closed edge-disjoint walk.
\item $C$ can be covered by a set of open edge-disjoint walks.
\end{enumerate}
\label{lemma:cover}
\end{lemma}

\begin{proof}
Lemma~\ref{lemma:conv-div} states that for non-empty connected components, the
component is either balanced or contains at least one divergent vertex. If the
component is balanced, the in-degree of each vertex is equal to the out-degree,
hence it is always possible to construct an Eulerian circuit in it using
Hierholzer's algorithm \cite{fleischner91}.  An Eulerian circuit is a closed
edge-disjoint walk by definition, thus $C$ satisfies case~1.

If $C$ is not balanced, there exists at least one divergent vertex in $C$.
We then construct a set of open walks using the following algorithm:

\begin{enumerate}
\item Select an arbitrary divergent vertex $v$. If there are no divergent
      vertices in the component, go to step 4.
\item Build a walk by following an arbitrary outgoing edge of the current
      vertex repeatedly until the walk gets stuck in a vertex $w$, while
	  making sure that each edge is included in the walk only once.
\item Store the walk, remove its edges from the component and go back to step
      1. Note that the walk is always open ($v \ne w$) since $v$ has
      more outbound edges than inbound ones, hence the walk cannot get
	  stuck in $v$.
\item At this step, there are no more divergent vertices in $C$. By
      Lemma~\ref{lemma:conv-div}, this
      implies that all the vertices are balanced. Since $C$ may have
      fallen apart into multiple connected components after the edge
      removals, construct an Eulerian circuit for each sub-component of
      $C$ and store it as a closed walk.
\end{enumerate}

The above algorithm provides us with a cover of $C$ with edge-disjoint
open and closed walks. However, note that each closed walk can be eliminated
by finding an open walk with which it shares a vertex $v$ and joining them
together in a larger open walk which traverses the original open walk from
the beginning until it arrives at $v$, then traverses the closes walk, and
resumes the original open walk at $v$ again. Repeating this procedure for
every closed walk in the cover provides us with a final cover containing
open walks only. This corresponds to case~2 in the lemma and concludes our
proof.
\end{proof}

Our key result is then as follows:

\begin{theorem}
The minimum set of driver nodes required to control the switchboard dynamics
on a network $G(V, E)$ can be determined by selecting the divergent vertices of
$G$ and one arbitrary vertex from each balanced component.
\label{th:main}
\end{theorem}

\begin{proof}
The proof will proceed as follows. First, we provide an algorithm which
constructs an edge cover in $G$ such that each open walk originates in a
divergent node and each balanced component is covered by a single closed walk,
giving an upper bound on the minimum number of driver nodes. Next, we show that
every divergent node must be driven in $G$ in any control configuration, and
that one arbitrary vertex from each balanced component must also be driven,
providing a lower bound on the minimum number of driver nodes.  We then show
that the upper and lower bounds coincide, therefore our algorithm is optimal.

We have already shown in the main part of this manuscript that the switchboard
dynamics on $G$ is equivalent to a linear time-invariant dynamics on the nodes
$L(G)$, for which a set of driver and driven nodes can be determined using the
maximum matching theorem of Liu et al \cite{liu11}. The maximum matching
theorem states that a given matching in $L(G)^\ast$ yields a set of stems
(directed vertex-disjoint paths) and buds (directed vertex-disjoint cycles) in
$L(G)$, and the roots of the stems (i.e.~the first vertices in the order of
traversal) have to be controlled by external signals\footnote{In case of a
non-maximum matching, some of the nodes have no incident edges selected in the
matching; these nodes can be considered as stems on their own.}. Buds do not
require separate driver nodes because they are either adjacent to a stem (and
thus use the signal from the stem) or one of the nodes in the bud is connected
to an already existing input signal directly. The driven nodes will be the
roots of the stems and an arbitrarily chosen vertex in each bud that is not
adjacent to a stem.

Each non-loop edge in the line digraph $L(G)$ corresponds to a length-two path
in $G$.  This implies that each stem in $L(G)$ corresponds to a concatenation
of length-two paths, yielding an edge-disjoint walk on $G$, which may
contain the same vertex multiple times but may not traverse the same edge
twice. Similarly, buds not containing a loop edge in $L(G)$ correspond to
edge-disjoint closed walks on $G$, and buds consisting of a single loop edge
in $L(G)$ yield a single open path of length 1 in $G$. Since each vertex in
$L(G)$ participates in either a stem or a bud (but not both at the same time),
mapping the stems and buds in $L(G)$ back to $G$ provides us with a cover of
$G$ using edge-disjoint closed and open walks. Note that the mapping is
injective: an edge-disjoint walk in $G$ can also be mapped back uniquely to a
stem or a bud in $L(G)$. Therefore, a matching in $L(G)$ is completely
equivalent to a cover of $G$ with edge-disjoint walks, and we are free to work
with either of them.

A possible cover of edge-disjoint walks can be obtained using the algorithm
described in Lemma~\ref{lemma:cover}. Such a cover creates a closed walk for
every balanced component and a set of open walks for every non-balanced
component. Mapping the walks to $L(G)$ gives us a set of stems and buds:

\begin{itemize}
\item Closed walks will become buds without loop edges in $L(G)$.
\item Open walks with at least two edges become stems in $L(G)$.
\item Open walks with a single edge will correspond to the appropriate
      loop edge in $L(G)$, thus becoming a bud with a single loop edge.
\end{itemize}

Together, these stems and buds form a set of control paths. Each stem requires
an input signal, hence the first vertex of each open walk in $G$ (i.e.~every
divergent vertex) will have to be driven. Since closed walks occur exclusively
within balanced components, and each balanced component contains only one
closed walk, the buds corresponding to them will not be adjacent to any of
the stems in $L(G)$, hence they will also have to be connected to some input
signal directly.  The only way we are allowed to achieve this in case of the
switchboard dynamics is to promote one of the nodes in the bud to a driver
node. Therefore, one driver node will be required for every balanced component
and for every divergent node of $G$. We have thus obtained an upper bound on
the number of driver nodes in $G$. To prove that the algorithm in
Lemma~\ref{lemma:cover} is optimal and conclude the proof, we will show that
this is also a lower bound.

Assume that there exists a complete cover of the edges (i.e. a control
configuration) of $G$ and there exists a divergent node $v$ such that $v$ is
not a driver node.  Since $v$ is not a driver node, there is no open walk
originating from it. Let us now consider all the walks $v$ is a part of. Closed
walks enter and leave $v$ the same number of times. Since no open walk
originates from $v$, each open walk either enters and leaves $v$ the same
number of times, or terminates in $v$.  Therefore, the number of covered
inbound edges of $v$ must be equal to or larger than the number of covered
outbound edges of $v$. However, since $v$ is divergent, it has more outbound
edges than inbound edges, therefore at least one outbound edge is not covered.
This contradicts our assumption that we are working with a complete cover.
Therefore, by contradiction, we have shown that every divergent vertex of $G$
must be a driver node in any control configuration.

Due to the reachability condition (see page~\pageref{reachability-cond}) , we
must also ensure that there is at least one driver node in every connected
component not containing a divergent vertex.  Lemma~\ref{lemma:conv-div} states
that every connected component $C$ of $G$ is empty, balanced, or contains at
least one convergent vertex. To satisfy the reachability condition, we must
therefore promote one of the vertices in every balanced component to a driver
node. Therefore, a lower bound on the number of driver nodes in $G$ is the
number of divergent nodes plus the number of balanced components of $G$. Since
the lower and upper bounds coincide, our algorithm is optimal. This concludes
our proof.
\end{proof}

Note that the algorithm given above allows one to determine the minimum set of
driver nodes for the switchboard dynamics on an arbitrary graph $G(V, E)$ in
$O(n+m)$ time (where $n$ is the number of vertices and $m$ is the number of
edges): building the edge-disjoint walks takes $O(m)$ time (since each edge has
to be evaluated only once), calculating the connected components takes
$O(n+m)$, and an additional $O(n)$ step can decide which connected components
are balanced.

\section{Analytical results}

In the main part of this manuscript, we have presented analytical formulae for
the expected fraction of driver nodes in Erd\H{o}s--R\'enyi, exponential and
scale-free networks. These formulae are based on the fact that the fraction of
driver nodes depends almost completely on the joint degree distribution of the
network according to Theorem~\ref{th:main}. By neglecting the possible
existence of balanced components, the fraction of driver nodes for graphs with
a joint degree distribution $\Pr(d_v^- = i, d_v^+ = j) = p_{ij}$ is simply
given by
\begin{equation}
n_D = \sum_{i=0}^\infty \sum_{j=i+1}^\infty p_{ij}
\label{eq:nd-general-1}
\end{equation}
i.e.~one simply has to calculate the sum of joint probabilities for cases when the
in-degree is smaller than the out-degree. When the in- and out-degrees are uncorrelated
and identically distributed (as in all the model networks we have presented in
the main part of this manuscript), it is also true that $p_{ij} = p_{ji}$,
hence $n_D$ can also be written as
\begin{equation}
n_D = \frac{1 - \sum_{k=0}^\infty p_{kk}}{2}
\label{eq:nd-general-2}
\end{equation}
i.e.~the fraction of driver nodes is equal to half the probability of non-balanced
nodes. The formulae presented in the main part of this manuscript are all based
on Eq.~\ref{eq:nd-general-2} by substituting the actual degree distribution of
the network model in question.

\subsection{Erdős--Rényi digraphs}

For Erd\H{o}s--R\'enyi digraphs with $n$ vertices and an edge probability of $p$,
both the in- and out-degrees follow a Poisson distribution with $\left< k
\right> = np$, hence $n_D$ is given as follows:
\begin{equation}
n_D^\mscript{ER} =
\frac{1}{2} \left( 1 - \sum_{k=0}^\infty \frac{\left<k\right>^{2k}}{k!k!} e^{-2\left<k\right>} \right) =
\frac{1}{2} \left( 1 - e^{-2\left<k\right>} I_0(2 \left<k\right>) \right)
\end{equation}
where $I_\alpha(x)$ is the modified Bessel function of the first kind. An equivalent
derivation follows from the fact that the difference of the in- and out-degree
of a node follows a Skellam distribution \cite{skellam46}, thus the probability
of balanced nodes is equal to the value of the probability moment function of
$\mbox{Skellam}(\left<k\right>, \left<k\right>)$ at $x=0$.

\subsection{Exponential networks}

In exponential networks, in-degrees and out-degrees are assumed to be distributed
with $\Pr(d_v^+ = k) = \Pr(d_v^- = k) = C e^{-k / \kappa}$ where
$C = 1 - e^{-1 / \kappa}$ and $\kappa = 1 / \log \frac{1 + \left< k \right>}
{\left< k \right>}$. The expected value of $n_D$ then follows from simple
algebraic manipulations:
\begin{eqnarray}
n_D^\mscript{exp} & = &
\frac{1}{2} \left( 1 - C^2 \sum_{i=0}^\infty e^{-2i / \kappa} \right) =
\frac{1}{2} \left( 1 - C^2 \frac{1}{1 - e^{-2 / \kappa}} \right) =
\frac{1}{2} \left( 1 - \frac{1 - e^{-1 / \kappa}}{1 + e^{-1 / \kappa}} \right) \\
& = & \frac{e^{-1/\kappa}}{1 + e^{-1/\kappa}} = \frac{\left<k\right>}{\left<k\right>+1}
\frac{\left<k\right> + 1}{2\left<k\right>+1} = 
\frac{\left< k \right>}{2 \left< k \right> + 1}
\end{eqnarray}
where in the penultimate step we have made use of $e^{-1/\kappa} = \frac{\left<k\right>}{\left<k\right>+1}$.

\subsection{Power-law networks}

In this case, we distinguish between networks with a power-law-like distribution that
has an exponential cutoff of the form $\Pr(d_v^+ = k) = \Pr(d_v^- = k) = C k^{-\gamma}
e^{-k/\kappa}$, and pure power-law distributions without a cutoff that follow
$\Pr(d_v^+ = k) = \Pr(d_v^- = k) = C k^{-\gamma}$. The exponential cutoff makes it
possible to normalize the distribution for a given average degree. We will start
with the former case and then show how $n_D$ behaves as the exponential cutoff
vanishes (i.e. $\kappa \to \infty$).

In the general case, $n_D$ is given by
\begin{equation}
n_D^\mscript{power} = \frac{1}{2} \left( 1 - C^2 \sum_{i=0}^\infty k^{-2\gamma} e^{-2k/\kappa} \right) = \frac{1}{2} \left( 1 - C^2 \Li_{2\gamma}(e^{-2/\kappa}) \right) 
= \frac{1}{2} - \frac{\Li_{2\gamma}(e^{-2/\kappa})}{2 \Li_\gamma(e^{-1/\kappa})^2}
\end{equation}
since we know that $C = \Li_\gamma(e^{-1/\kappa})$, where $\Li_s(z)$ is the
polylogarithm function. For $z=1$, the polylogarithm reduces to the Riemann
zeta function, yielding
\begin{equation}
n_D^\mscript{power} = \frac{1}{2} - \frac{\zeta(2\gamma)}{2 \zeta(\gamma)^2}
\end{equation}
for pure power-law networks.

\subsection{$k$-regular networks}

The three network models presented so far produce balanced components with a
very low probability, hence we were safe to ignore such components in our
analytical calculations. In this section, we present similar calculations
for networks where the in- and out-degree of each vertex is $k/2$ for some
even $k$. These networks consist of balanced nodes only, and the number of
driver nodes is given by the number of connected components of the graph containing
at least one edge.

\begin{theorem}
In a $k$-regular directed network $G(V, E)$ with $n$ vertices, the number of
driver nodes is zero if $k = 0$, one if $k \ge 4$ and the $n$th harmonic
number $H_n$ if $k = 2$.
\end{theorem}

\begin{proof}
The case of $k = 0$ is trivial: there are no edges to control and hence the
fraction of driver nodes is zero. For $k \ge 4$, dropping the arrowheads gives
us an undirected $k$-regular graph where it can be proven that it is almost
surely $k$-connected \cite{bollobas01}, implying that the original digraph
requires only one driver node. For $k = 2$, each vertex has exactly one inbound
and one outbound edge, thus the entire graph consists of disjoint directed
cycles. By denoting the head of the outbound edge of vertex $v$ by $\pi(v)$, we
obtain a permutation $\pi$ on the vertices of the graph, and the number of
connected components will be given by the number of cycles in $\pi$.

Let us call a sequence of elements $u_1, u_2, \dots, u_m$ an $m$-cycle of $\pi$
if each $u_i$ is equal to some $v_j$ and it holds that $\pi(u_1) = u_2,
\pi(u_2) = u_3, \dots, \pi(u_m) = u_1$.
First we prove that the probability of the event that $v_1$ is a part of an
$m$-cycle is $1/n$. We require that $u_1 = v_1$, $u_2 = \pi(u_1) \ne v_1$,
$u_3 = \pi(u_2) \ne v_1, \dots \pi(u_m) = v_1$. Therefore,

$$
\Pr(v_1\mbox{~is in an $m$-cycle}) = \frac{n-1}{n} \frac{n-2}{n-1} \frac{n-3}{n-2} \cdots \frac{n-m}{n-m+1}\frac{1}{n-m} = \frac{1}{n}
$$

Of course the above proof applies to every $v_i \in V$. Since each $v_i$ is a
part of an $m$-cycle with probability $1/n$, the expected number of vertices
being part of an $m$-cycle is exactly 1, and since an $m$-cycle contains $m$
vertices, the expected number of $m$-cycles is $1/m$. The expected number of
cycles of any length $\mathcal{E}$ then follows by a simple summation:

$$
\mathcal{E} = \sum_{m=1}^n \frac{1}{m} = H_n
$$

This concludes our proof.
\end{proof}

Since $H_n$ scales approximately as $\log n$, the fraction of driver nodes will
scale as $\log n / n$ and tend to zero as $n \to \infty$.  $k$-regular graphs
are thus extremely well-controllable in the infinite limit, requiring $O(1)$
driver nodes if $k \ne 2$ and $O(\log n)$ driver nodes if $k = 2$.

\section{Computational results}

\subsection{Data sources of real networks}

\label{sec:reversal}
The details of the real-world networks we have studied are presented in
Table~\ref{tab:real}. Note that the semantics of the switchboard
dynamics requires that a directed $A \to B$ edge represents a direct influence of
$A$ on $B$ and not the other way round, hence we had to reverse the edge directions
in some of the networks to make it conform to this semantics. For instance, an
$A \to B$ edge in a trust network usually means that $A$ trusts $B$, hence $B$
has a direct influence on $A$. For sake of clarity, the table includes the
semantics of each edge.

\begin{table}
\begin{center}
\footnotesize
\begin{tabular}{p{2.2cm}rccclp{4cm}} \hline
Type & \# & Name & $n$ & $m$ & & Semantics of $A \to B$ \\[2pt] \hline \\[-5pt]
Regulatory & 1. & Ownership-USCorp & 7,253 & 6,726 & \cite{norlen02} & $A$ owns $B$ \\
& 2. & TRN-EC-2 & 418 & 519 & \cite{milo02} & $A$ regulates $B$ \\
& 3. & TRN-Yeast-1 & 4,441 & 12,873 & \cite{balaji06} & $A$ regulates $B$ \\
& 4. & TRN-Yeast-2 & 688 & 1,079 & \cite{milo02} & $A$ regulates $B$ \\[5pt]

Trust & 5. & College$^\ast$ & 32 & 96 & \cite{vanduijn03,milo04} & $A$ is trusted by $B$ \\
& 6. & Epinions$^\ast$ & 75,888 & 508,837 & \cite{richardson03} & $A$ is trusted by $B$ \\
& 7. & Prison$^\ast$ & 67 & 182 & \cite{vanduijn03,milo04} & $A$ is trusted by $B$ \\
& 8. & Slashdot$^\ast$ & 82,168 & 948,464 & \cite{leskovec09} & $A$ is trusted by $B$ \\
& 9. & WikiVote$^\ast$ & 7,115 & 103,689 & \cite{leskovec10} & $A$ was voted on by $B$ \\[5pt]

Food web & 10. & Grassland & 88 & 137 & \cite{dunne02} & $A$ preys on $B$ \\
& 11. & Little Rock & 183 & 2,494 & \cite{martinez91} & $A$ preys on $B$ \\
& 12. & Seagrass & 49 & 226 & \cite{christian99} & $A$ preys on $B$ \\
& 13. & Ythan & 135 & 601 & \cite{dunne02} & $A$ preys on $B$ \\[5pt]

Metabolic & 14. & \emph{C. elegans} & 1,173 & 2,864 & \cite{jeong00} & $B$ is produced from $A$ \\
& 15. & \emph{E. coli} & 2,275 & 5,763 & \cite{jeong00} & $B$ is produced from $A$ \\
& 16. & \emph{S. cerevisiae} & 1,511 & 3,833 & \cite{jeong00} & $B$ is produced from $A$ \\[5pt]

Electronic & 17. & s208a & 122 & 189 & \cite{milo02} & $B$ is a function of $A$ \\
circuits & 18. & s420a & 252 & 399 & \cite{milo02} & $B$ is a function of $A$ \\
 & 19. & s838a & 512 & 819 & \cite{milo02} & $B$ is a function of $A$ \\[5pt]

Neuronal and brain & 20. & \emph{C. elegans} & 297 & 2,359 & \cite{achacoso92,watts98} & $B$ is within one synapse or gap junction distance from $A$ \\
& 21. & Macaque & 45 & 463 & \cite{negyessy06} & Area $A$ is connected to area $B$ \\[5pt]

Citation & 22. & arXiv-HepPh$^\ast$ & 34,546 & 421,578 & \cite{leskovec05} & $A$ is cited by $B$ \\
& 23. & arXiv-HepTh$^\ast$ & 27,770 & 352,807 & \cite{leskovec05} & $A$ is cited by $B$ \\[5pt]

WWW & 24. & Google & 15,763 & 171,206 & \cite{palla07} & $A$ links to $B$ \\
& 25. & Polblogs & 1,490 & 19,090 & \cite{adamic05} & $A$ links to $B$ \\
& 26. & nd.edu & 325,729 & 1,497,134 & \cite{albert99} & $A$ links to $B$ \\
& 27. & stanford.edu & 281,904 & 2,312,497 & \cite{leskovec05} & $A$ links to $B$ \\[5pt]

Internet & 28. & p2p-1 & 10,876 & 39,994 & \cite{ripeanu02,leskovec05} & $A$ sent messages to $B$ \\
 & 29. & p2p-2 & 8,846 & 31,839 & \cite{ripeanu02,leskovec05} & $A$ sent messages to $B$ \\
 & 30. & p2p-3 & 8,717 & 31,525 & \cite{ripeanu02,leskovec05} & $A$ sent messages to $B$ \\[5pt]

Social & 31. & Twitter$^\ast$ & $41.7 \times 10^6$ & $1.47 \times 10^9$  & \cite{kwak10} & $A$ is followed by $B$ \\
communication & 32. & UCIOnline & 1,899 & 20,296 & \cite{panzarasa09} & $A$ sent emails to $B$ \\
 & 33. & WikiTalk & 2,394,385 & 5,021,410 & \cite{leskovec10} & $A$ edited the talk page of $B$ on Wikipedia \\[5pt]

Intra- & 34. & Consulting$^\ast$ & 46 & 879 & \cite{cross04} & $B$ turned to $A$ for advice \\
organizational & 35. & Freemans-1$^\ast$ & 34 & 645 & \cite{freeman79} & $A$ was nominated by $B$ on a questionnare as acquaintance \\
& 36. & Freemans-2$^\ast$ & 34 & 830 & \cite{freeman79} & $A$ was nominated by $B$ on a questionnare as acquaintance \\
& 37. & Manufacturing$^\ast$ & 77 & 2,228 & \cite{cross04} & $B$ turned to $A$ for advice \\
& 38. & University$^\ast$ & 81 & 817 & \cite{nepusz08} & $A$ was nominated by $B$ on a questionnare \\[5pt]

\hline
\end{tabular}
\caption{Summary of the real networks analyzed in the paper. $n$ denotes the number of
nodes, $m$ denotes the number of edges. Networks where the edges were reversed
compared to the original publication are marked by an asterisk ($\ast$).}
\end{center}
\label{tab:real}
\end{table}

\subsection{Robustness of control configurations}

To study the robustness of real networks against random control path failures,
we have classified each edge according to the change in the number of driver
nodes when the edge is removed from the network. We distinguish three cases and
accordingly three classes of edges. The removal of a \emph{critical} edge
increases the number of driver nodes required to maintain controllability.
Conversely, the removal of a so-called \emph{distinguished} edge decreases the
number of driver nodes.  The remaining edges are called \emph{ordinary} since
their removal does not affect the set of driver nodes.

Figure~\ref{fig:edge_classes} shows the fraction of critical, ordinary and
distinguished edges in each studied real network, indicating that most networks
possess only a small fraction of critical or distinguished edges, thus exhibiting
a high degree of robustness against changes in control configurations due to
random edge removals. The two significant exceptions are the electronic circuit
networks (s208a, s420a and s838a) \cite{milo02}, which contain a high fraction
of distinguished edges, and the metabolic networks \cite{jeong00}, where almost
half of the edges are critical.

\begin{figure}
\includegraphics[width=\textwidth]{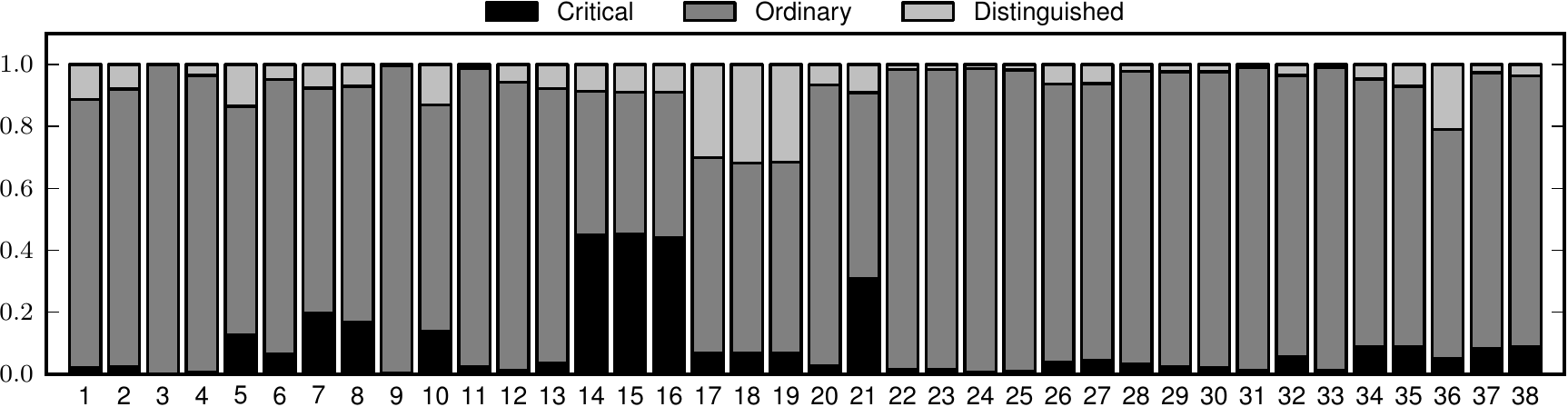}
\caption{Fraction of distinguished (light gray), ordinary (dark gray) and
critical (black) edges in the real networks studied in this paper. Numbers
refer to the network indices in Table~\ref{tab:real}.}
\label{fig:edge_classes}
\end{figure}

\subsection{Implementation}

An open-source implementation of the driver node calculations and the edge
classification for arbitrary networks is provided at
\texttt{http://github.com/ntamas/netctrl}.


\begin{thebibliography}{10}

\bibitem{achacoso92}
T.~Achacoso and W.~Yamamoto.
\newblock {\em {AY}'s Neuroanatomy of C. elegans for Computation}.
\newblock CRC Press, Boca Raton, FL, 1st edition, 1992.

\bibitem{adamic05}
L.~Adamic and N.~Glance.
\newblock The political blogosphere and the 2004 {US} election.
\newblock In {\em Proceedings of the WWW-2005 Workshop on the Weblogging
  Ecosystem}. 2005.

\bibitem{albert02}
R.~Albert and A.-L. Barab\'asi.
\newblock Statistical mechanics of complex networks.
\newblock {\em Rev. Mod. Phys.}, 74:47, 2002.

\bibitem{albert99}
R.~Albert, H.~Jeong, and A.-L. Barab\'asi.
\newblock Diameter of the {W}orld {W}ide {W}eb.
\newblock {\em Nature}, 401:130--131, 1999.

\bibitem{albert00}
R.~Albert, H.~Jeong, and A.-L. Barab\'asi.
\newblock Error and attack tolerance of complet networks.
\newblock {\em Nature}, 406:378--382, 2000.

\bibitem{amaral00}
L.~Amaral, A.~Scala, M.~Barthel\'emy, and H.~Stanley.
\newblock Classes of small-world networks.
\newblock {\em Proc. Natl. Acad. Sci. USA}, 97:11149--11152, 2000.

\bibitem{balaji06}
S.~Balaji, M.~Babu, L.~Iyer, N.~Luscombe, and L.~Aravind.
\newblock Comprehensive analysis of combinatorial regulation using the
  transcriptional regulatory network of yeast.
\newblock {\em J. Mol. Biol}, 360(1):213--27, 2006.

\bibitem{barabasi99}
A.-L. Barab\'asi and R.~Albert.
\newblock Emergence of scaling in random networks.
\newblock {\em Science}, 286:509--512, 1999.

\bibitem{bianconi08}
G.~Bianconi, N.~Gulbahce, and A.~Motter.
\newblock Local structure of directed networks.
\newblock {\em Phys. Rev. Lett.}, 100:118701, 2008.

\bibitem{bianconi05}
G.~Bianconi and M.~Marsili.
\newblock Loops of any size and hamilton cycles in random scale-free networks.
\newblock {\em J. Stat. Mech.}, page P06005, 2005.

\bibitem{boccaletti06}
S.~Boccaletti, V.~Latora, Y.~Moreno, M.~Chavez, and D.-U. Huang.
\newblock Complex networks: structure and dynamics.
\newblock {\em Phys. Rep.}, 424(4--5):175--308, 2006.

\bibitem{bollobas01}
B.~Bollob\'as.
\newblock {\em Random Graphs}.
\newblock Cambridge Studies in Advanced Mathematics. Cambridge University
  Press, Cambridge, second edition, 2001.

\bibitem{caldarelli07}
G.~Caldarelli.
\newblock {\em Scale-Free Networks: Complex Web in Nature and Technology}.
\newblock Oxford University Press, 2007.

\bibitem{christian99}
R.~Christian and J.~Luczkovich.
\newblock Organizing and understanding a winter's seagrass foodweb network
  through effective trophic levels.
\newblock {\em Ecological Modelling}, 117:99--124, 1999.

\bibitem{cohen00}
R.~Cohen, K.~Erez, D.~Ben-Avraham, and S.~Havlin.
\newblock Resilience of the {I}nternet to random breakdowns.
\newblock {\em Phys. Rev. Lett.}, 85:4626--4628, 2000.

\bibitem{cowan11}
N.~Cowan, E.~Chastain, D.~Vilhena, J.~Freudenberg, and C.~Bergstrom.
\newblock Nodal dynamics determine the controllability of complex networks,
  2011.

\bibitem{cross04}
R.~Cross and A.~Parker.
\newblock {\em The Hidden Power of Social Networks}.
\newblock Harvard Business School Press, Boston, MA, USA, 2004.

\bibitem{dunne02}
J.~Dunne, R.~Williams, and N.~Martinez.
\newblock Food-web structure and network theory: {T}he role of connectance and
  size.
\newblock {\em Proc. Natl. Acad. Sci. USA}, 99(20):12917--22, 2002.

\bibitem{ebel02}
H.~Ebel, J.~Davidsen, and S.~Bornholdt.
\newblock Dynamics of social networks.
\newblock {\em Complexity}, 8:24--27, 2002.

\bibitem{erdos60}
P.~Erd\H{o}s and A.~R\'enyi.
\newblock On the evolution of random graphs.
\newblock {\em Publ. Math. Inst. Hung. Acad. Sci.}, 5:17--60, 1960.

\bibitem{fleischner91}
H.~Fleischner.
\newblock {\em Algorithms for Eulerian trails}, volume~50 of {\em Annals of
  Discrete Mathematics}.
\newblock Elsevier, 1991.

\bibitem{fortunato10}
S.~Fortunato.
\newblock Community detection in graphs.
\newblock {\em Phys. Rep.}, 486:75--174, 2010.

\bibitem{freeman79}
S.~Freeman and L.~Freeman.
\newblock Social science research reports 46.
\newblock Technical report, University of California, Irvine, CA, 1979.

\bibitem{jeong01}
H.~Jeong, S.~Mason, A.-L. Barab\'asi, and Z.~Oltvai.
\newblock Lethality and centrality in protein networks.
\newblock {\em Nature}, 411:41--42, 2001.

\bibitem{jeong00}
H.~Jeong, B.~Tombor, R.~Albert, Z.~Oltvai, and A.~Barabási.
\newblock The large-scale organization of metabolic networks.
\newblock {\em Nature}, 407(6804):651--4, 2000.

\bibitem{kalman63}
R.~Kalman.
\newblock Mathematical description of linear dynamical systems.
\newblock {\em J. Soc. Indus. Appl. Math. Ser. A}, 1:152--192, 1963.

\bibitem{kwak10}
H.~Kwak, C.~Lee, H.~Park, and S.~Moon.
\newblock What is {T}witter, a social network or a news media?
\newblock In {\em WWW'10: Proceedings of the 19th International Conference on
  World Wide Web}, pages 591--600, Raleigh, North Carolina, USA, 2010. ACM.

\bibitem{leskovec05}
J.~Leskovec and C.~Faloutsos.
\newblock Graphs over time: densification laws, shrinking diameters and
  possible explanations.
\newblock In {\em Proceedings of the ACM SIGKDD International Conference on
  Knowledge Discovery and Data Mining (KDD)}, 2005.

\bibitem{leskovec10}
J.~Leskovec, D.~Huttenlocher, and J.~Kleinberg.
\newblock Signed networks in social media.
\newblock In {\em Proceedings of the ACM SIGCHI Conference on Human Factors in
  Computing Systems (CHI)}, 2010.

\bibitem{leskovec09}
J.~Leskovec, K.~Lang, A.~Dasgupta, and M.~Mahoney.
\newblock Community structure in large networks: natural cluster sizes and the
  absence of large well-defined clusters.
\newblock {\em Internet Mathematics}, 6(1):29--123, 2009.

\bibitem{lin74}
C.~Lin.
\newblock Structural controllability.
\newblock {\em IEEE Trans. Automat. Contr.}, 19:201--8, 1974.

\bibitem{liu11}
Y.~Liu, J.~Slotine, and A.~Barabási.
\newblock Controllability of complex networks.
\newblock {\em Nature}, 473(7346):167--73, 2011.

\bibitem{lombardi07}
A.~Lombardi and M.~H\"{o}rnquist.
\newblock Controllability analysis of networks.
\newblock {\em Phys Rev E}, 75:056110, 2007.

\bibitem{luscombe04}
N.~Luscombe, M.~Madan~Babu, H.~Yu, M.~Snyder, A.~Teichmann, and M.~Gerstein.
\newblock Genomic analysis of regulatory network dynamics reveals large
  topological changes.
\newblock {\em Nature}, 431:308--312, 2004.

\bibitem{martinez91}
N.~Martinez.
\newblock Artifacts or attributes? {E}ffects of resolution on the {L}ittle
  {R}ock {L}ake food web.
\newblock {\em Ecological Monographs}, 61:367--392, 1991.

\bibitem{milo04}
R.~Milo, S.~Itzkovitz, N.~Kashtan, R.~Levitt, S.~Shen-Orr, I.~Ayzenshtat,
  M.~Sheffer, and U.~Alon.
\newblock Superfamilies of evolved and designed networks.
\newblock {\em Science}, 303(5663):1538--42, 2004.

\bibitem{milo02}
R.~Milo, S.~Shen-Orr, S.~Itzkovitz, N.~Kashtan, D.~Chklovskii, and U.~Alon.
\newblock Network motifs: simple building blocks of complex networks.
\newblock {\em Science}, 298(5594):824--7, 2002.

\bibitem{murota87}
K.~Murota.
\newblock {\em Systems Analysis by Graphs and Matroids}.
\newblock Springer-Verlag, Berlin, 1987.

\bibitem{nagy10}
M.~Nagy, Z.~\'Akos, D.~Biro, and T.~Vicsek.
\newblock Hierarchical group dynamics in pigeon flocks.
\newblock {\em Nature}, 464:890--893, 2010.

\bibitem{negyessy06}
L.~N\'egyessy, T.~Nepusz, L.~Kocsis, and F.~Bazs\'o.
\newblock Prediction of the main cortical areas and connections involved in the
  tactile function of the visual cortex by network analysis.
\newblock {\em Eur. J. Neurosci.}, 23(7):1919--1930, 2006.

\bibitem{nepusz08}
T.~Nepusz, A.~Petr\'oczi, L.~N\'egyessy, and F.~Bazs\'o.
\newblock Fuzzy communities and the concept of bridgeness in complex networks.
\newblock {\em Phys. Rev. E}, 77:016107, 2008.

\bibitem{newman03}
M.~Newman.
\newblock The structure and function of complex networks.
\newblock {\em SIAM Rev.}, 45:167--256, 2003.

\bibitem{newman04}
M.~Newman and M.~Girvan.
\newblock Finding and evaluating community structure in networks.
\newblock {\em Phys. Rev. E}, 69:026113, 2004.

\bibitem{newman02b}
M.~Newman, D.~Watts, and S.~Strogatz.
\newblock Random graph models of social networks.
\newblock {\em Proc. Natl. Acad. Sci. USA}, 99(Suppl 1):2566--2572, 2002.

\bibitem{norlen02}
K.~Norlen, G.~Lucas, M.~Gebbie, and J.~Chuang.
\newblock {EVA}: {E}xtraction, visualization and analysis of the
  telecommunications and media ownership network.
\newblock In {\em Proceedings of the International Telecommunications Society
  14th Biennial Conference (ITS2002)}, Seoul, South Korea, August 2002.

\bibitem{palla07b}
G.~Palla, A.-L. Barab\'asi, and T.~Vicsek.
\newblock Quantifying social group evolution.
\newblock {\em Nature}, 446:664--667, 2007.

\bibitem{palla05}
G.~Palla, I.~Der\'enyi, I.~Farkas, and T.~Vicsek.
\newblock Uncovering the overlapping community structure of complex networks in
  nature and society.
\newblock {\em Nature}, 435:814--818, 2005.

\bibitem{palla07}
G.~Palla, I.~Farkas, P.~Pollner, I.~Der\'enyi, and T.~Vicsek.
\newblock Directed network modules.
\newblock {\em New J. Phys.}, 9:186, 2007.

\bibitem{panzarasa09}
P.~Panzarasa, T.~Opsahl, and K.~Carley.
\newblock Patterns and dynamics of users' behaviour and interaction: Network
  analysis of an online community.
\newblock {\em Journal of the American Society for Information Science and
  Technology}, 60(5):911--932, 2009.

\bibitem{pastorsatorras01}
R.~Pastor-Satorras and A.~Vespignani.
\newblock Epidemic spreading in scale-free networks.
\newblock {\em Phys. Rev. Lett.}, 86:3200--3203, 2001.

\bibitem{pastorsatorras02}
R.~Pastor-Satorras and A.~Vespignani.
\newblock Immunization of complex networks.
\newblock {\em Phys. Rev. E}, 65:036104, 2002.

\bibitem{prill05}
R.~Prill, P.~Iglesias, and A.~Levchenko.
\newblock Dynamic properties of network motifs contribute to biological network
  organization.
\newblock {\em PLoS Biol.}, 3(11):e343, 2005.

\bibitem{rahmani09}
A.~Rahmani, M.~Ji, M.~Mesbahi, and M.~Egerstedt.
\newblock Controllability of multi-agent systems from a graph-theoretic
  perspective.
\newblock {\em SIAM J Contr. Optim.}, 48:162--186, 2009.

\bibitem{reinschke97}
K.~Reinschke and G.~Wiedemann.
\newblock Digraph characterization of structural controllability for linear
  descriptor systems.
\newblock {\em Linear Algebra Appl.}, 266:199--217, 1997.

\bibitem{richardson03}
M.~Richardson, R.~Agrawal, and P.~Domingos.
\newblock Trust management for the semantic web.
\newblock In {\em Proceedings of the Second International Semantic Web
  Conference}, 2003.

\bibitem{ripeanu02}
M.~Ripeanu, I.~Foster, and A.~Iamnitchi.
\newblock Mapping the {G}nutella network: {P}roperties of large-scale
  peer-to-peer systems and implications for system design.
\newblock {\em IEEE Internet Computing Journal}, 6(1):50--57, 2002.

\bibitem{shields76}
R.~Shields and J.~Pearson.
\newblock Structural controllability of multi-input linear systems.
\newblock {\em IEEE Trans. Automat. Contr.}, 21:203--212, 1976.

\bibitem{skellam46}
J.~Skellam.
\newblock The frequency distribution of the difference between two poisson
  variables belonging to different populations.
\newblock {\em Journal of the Royal Statistical Society: Series A}, 109(3):296,
  1946.

\bibitem{slotine91}
J.-J. Slotine and W.~Li.
\newblock {\em Applied Nonlinear Control}.
\newblock Prentice-Hall, New Jersey, 1991.

\bibitem{sontag98}
E.~Sontag.
\newblock {\em Mathematical Control Theory}.
\newblock Springer, New York, 1998.

\bibitem{vanduijn03}
M.~A.~J. Van~Duijn, M.~Huisman, F.~N. Stokman, F.~W. Wasseur, and E.~P.~H.
  Zeggelink.
\newblock Evolution of sociology freshmen inito a friendship network.
\newblock {\em J. Math. Soc.}, 27:153--191, 2003.

\bibitem{watts98}
D.~Watts and S.~Strogatz.
\newblock Collective dynamics of 'small-world' networks.
\newblock {\em Nature}, 393(6684):440--2, 1998.

\bibitem{yu09}
W.~Yu, G.~Chen, and J.~L\"{u}.
\newblock On pinning synchronization of complex dynamical networks.
\newblock {\em Automatica}, 45:429--435, 2009.

\end{thebibliography}
\end{document}